%% file: main.tex
\documentclass{article}

\usepackage[margin=1in]{geometry}

\usepackage{authblk}
\usepackage{cite}
\usepackage{microtype}
\usepackage{xspace}
\usepackage{enumitem}
\usepackage[linesnumbered]{algorithm2e}
\usepackage[dvipsnames]{xcolor}
\usepackage[pdf]{pstricks} 
\usepackage{amsmath}
\usepackage{amsthm}
\usepackage{amssymb}

\usepackage{auto-pst-pdf}
\usepackage{subcaption}

\newcommand{\ceil}[1]{\lceil #1 \rceil}
\renewcommand{\epsilon}{\varepsilon}
\renewcommand{\emph}[1]{\textit{\textbf{#1}}}
\newcommand{\conv}[1]{conv(#1)}
\newcommand{\ext}[1]{ext(#1)}
\newcommand{\depth}[1]{\mathrm{depth}(#1)}

\newtheorem{theorem}{Theorem}

\newtheorem{lemma}[theorem]{Lemma}

\newtheorem{corollary}[theorem]{Corollary}
\newtheorem{observation}[theorem]{Observation}

\begin{document}

\input{title}

\thispagestyle{empty}

\begin{abstract}

  Many classical algorithms are known for computing the convex hull of a
  set of $n$ point in $\mathbb{R}^2$ using $O(n)$ space.  For large
  point sets, whose size exceeds the size of the working space, these
  algorithms cannot be directly used. The current best streaming
  algorithm for computing the convex hull is computationally expensive, 
	because it needs to solve a set of linear programs.

  In this paper, we propose simpler and faster streaming and W-stream
  algorithms for computing the convex hull. Our streaming algorithm
  has small pass complexity, which is roughly a square root of the
  current best bound, and it is simpler in the sense that our algorithm
	mainly relies on computing the convex hulls of smaller point sets. 
	Our W-stream algorithms, one of which
  is deterministic and the other of which is randomized, have
  nearly-optimal tradeoff between the pass complexity and space usage,
  as we established by a new unconditional lower bound.

\end{abstract}


\clearpage

\pagenumbering{arabic}

\input{intro}

\input{pre}

\input{algo}

\input{wrstr}

\input{lower}

\input{clower}

\input{proof}

\input{main.ref}
\end{document}

%% file: title.tex
\title{Streaming Algorithms for Planar Convex Hulls}


\author[1]{Mart\'{\i}n Farach-Colton\footnote{
This research was supported in part by NSF CCF 1637458, NIH 1 U01 CA198952-01, a NetAPP Faculty Fellowship and a gift from Dell/EMC.
}}
\author[1]{Meng Li}
\author[2]{Meng-Tsung Tsai\footnote{
This research was supported in part by the Ministry of Science and Technology of Taiwan under contract MOST grant 107-2218-E-009- 026-MY3, and the Higher Education Sprout Project of National Chiao Tung University and Ministry of Education (MOE), Taiwan.
}}

\affil[1]{Rutgers University, Piscataway, USA \protect\\
  \texttt{\{farach, ml910\}@cs.rutgers.edu} }
\affil[2]{National Chiao Tung University, Hsinchu, Taiwan \protect\\
	\texttt{mtsai@cs.nctu.edu.tw} }

\date{}

\maketitle

%% file: intro.tex
\section{Introduction}\label{sec:intro}

The \emph{convex hull} of a set $P$ of points in $\mathbb{R}^2$ is the smallest convex set that contains $P$. We denote the convex hull of $P$ by $\conv{P}$ and denote the set of extreme points in $\conv{P}$ by $\ext{P}$. Let $n = |P|$ and $h = |\ext{P}|$. Note that $h \le n$ because $\ext{P}$ is a subset of $P$. By computing the convex hull of $P$, we mean outputting the points in $\ext{P}$ in clockwise order.

There is a long line of research on computing the convex hull using $O(n)$ space. In the RAM model, Graham~\cite{g72} gave the first algorithm, called the {\it Graham Scan}, with running time $O(n \log n)$. Subsequently, several algorithms were devised with the same running time, but with different approaches~\cite{ph77,a79,k84,bdh96}. In the output-sensitive model, where the running time depends on $n$ and $h$,
Jarvis~\cite{j73} proposed the {\it Gift Wrapping} algorithm, which has running time $O(nh)$. This algorithm was later improved by Kirkpatrick and Seidel~\cite{ks86} and Chan~\cite{c96}, both of which achieve  running time of $O(n \log h)$. In the online model, where input points are given one by one and algorithms need to compute the convex hull of points seen so far, Overmars and van Leeuween's algorithm~\cite{ol81} can update the convex hull in $O(\log^2 n)$ time per incoming point. Brodal and Jacob~\cite{bj02} reduced the update time to $O(\log n)$.

\noindent\textbf{Streaming Model.} The algorithms mentioned above all require $s = \Omega(n)$ working space (memory) in the worst case. Therefore, none of these can handle the case when $s \ll n$, that is, when either $n$ is very large (a massive data set) or $s$ is very small (such as in embedded systems). In order to explore the convex hull problem with such a memory restriction, we consider the standard streaming models~\cite{ruhlthesis03,bbdmw02,cia06,muthu06,cbga10}, where the $n$ given points are stored on a read-only or writable tape in an arbitrary order. If the tape is read-only, then the model is simply called the \emph{streaming model}~\cite{bbdmw02,muthu06}. Otherwise the tape is writable, and the model is called the \emph{W-stream model}~\cite{ruhlthesis03,cia06,cbga10}. We refer to algorithms in the streaming model as \emph{streaming algorithms} and algorithms in the W-stream model as \emph{W-stream algorithms}.
In both models, algorithms can manipulate the working space while reading the points sequentially from the beginning of the tape to the end; however, only algorithms in the W-stream model can modify the tape, detailed in Section~\ref{sec:wrstr}. Hence, algorithms in this model cannot access the input randomly, which is different from the model for in-place algorithms~\cite{Bronn04,Bose07}.
The extreme points are written to a write-only stream. The \emph{pass complexity} of an algorithm refers to the number of times the algorithms scans the tape from the beginning to the end. The goal is to devise streaming and W-stream algorithms that have small pass and space complexities. 

No single-pass streaming algorithm can compute the convex hull using $o(n)$ space because it is no easier than sorting $n$ positive numbers in $\mathbb{R}$. Since sorting $n$ numbers using $s$ spaces requires $\Omega(n/s)$ passes~\cite{mp78}, computing the convex hull in a single pass requires linear space.  However, Chan and Chen~\cite{Chan07} showed that the space requirement can be significantly reduced if multi-pass algorithms are allowed. Specifically, their streaming algorithm uses $O(\delta^{-2})$ 
passes, $O(\delta^{-2} h n^\delta)$ space, and $O(\delta^{-2} n \log n)$ time for any constant $\delta \in (0, 1)$. On the other hand, to have small space complexity, one can appeal to a general scheme to convert PRAM algorithms to W-stream algorithms established by Demetrescu et al.~\cite{cbga10}, summarized in Section~\ref{sec:wrstr}. Using this technique yields a W-stream algorithm that uses $O((n/s) \log h)$ 
passes and $O(s)$ space where $s$ can be as small as constant. 

\smallskip

\noindent\textbf{Our Contribution.} We devise a new $O(n \log h)$-time RAM algorithm to compute the convex hull (Section~\ref{sec:pre}). Then, we adapt the RAM algorithm to both models. 

In the streaming model, the pass complexity of our algorithm is roughly a square root of that of Chan and Chen's algorithm~\cite{Chan07} if both algorithms have the same space usage. Specifically, we have:

\begin{theorem}\label{thm:read-only}
Given a set $P$ of $n$ points in $\mathbb{R}^2$ on a read-only tape where $|\ext{P}| = h$, there exists a deterministic streaming algorithm to compute the convex hull of $P$ in $O(\delta^{-1})$ passes using $O(\min\{\delta^{-1} h n^\delta \log n, n\})$ space and $O(\delta^{-2} n \log n)$ time for every constant $\delta \in (0, 1)$.
\end{theorem}

In the W-stream model, we adapt the RAM algorithm to two W-stream algorithms. One uses $O(s)$ space for any $s = \Omega(\log n)$ and the other uses $O(s)$ space for any $s = \Omega(1)$. The pass complexity of our W-stream algorithms are $O(\ceil{h/s}\log n)$ and $O(h/s + \log n)$, which are smaller than $O((n/s)\log h)$, the best pass complexity among those W-stream algorithms that are converted from PRAM algorithms in algebraic decision tree model~\cite{cbga10}, when $s \le h$.  

The first W-stream algorithm is deterministic, and we get:  

\begin{theorem}\label{thm:det}
Given a set $P$ of $n$ points in $\mathbb{R}^2$ where $|\ext{P}| = h$, there exists a deterministic W-stream algorithm to compute the convex hull of $P$ in $O(\lceil h/s \rceil \log n)$ passes using $O(s)$ space and $O(n\log^2 n)$ time 
for any $s = \Omega(\log n)$.
\end{theorem} 

Next, we randomize the above W-stream algorithm. A logarithmic factor can be shaved off from the pass complexity w.h.p.\footnote{w.h.p. means with probability $1-1/n^{\Omega(1)}$.} We have:

\begin{theorem}\label{thm:rand}
Given a set $P$ of $n$ points in $\mathbb{R}^2$ where $|\ext{P}| = h$, there exists a randomized W-stream algorithm to compute the convex hull of $P$ in $p$ passes using $O(s)$ space and $O(n\log^2 n)$ time 
for any $s = \Omega(1)$, where $p = O(h/s + \log n)$ w.h.p.
\end{theorem} 

We prove that our W-stream algorithms have nearly-optimal tradeoff between pass and space complexities by showing Theorem~\ref{thm:lower}, which generalizes Guha and McGregor's lower bound (Theorem~8 in~\cite{Guha08}). We remark that this lower bound is sharp because it matches the bounds of our randomized W-stream algorithm when $h = \Omega(s \log n)$. 

\begin{theorem}\label{thm:lower}
Given a set $P$ of $n$ points in $\mathbb{R}^2$ where $|\ext{P}| = h = \Omega(1)$, any streaming (or W-stream) algorithm that computes the convex hull of $P$ with success rate $\ge 2/3$, and uses  $s$ bits requires $\Omega(\ceil{h/s})$ passes. 
\end{theorem}

We note here that space is measured in terms of bits for lower bounds and in terms of points for upper bounds. This asymmetry is a common issue for geometric problems because most geometric problems are analyzed under the RealRAM model, where precision of points (or other geometric objects) is unbounded.

\smallskip

\noindent\textbf{Applications.} Our W-stream algorithms can handle the case for $s \le h$ because it outputs extreme points on the fly.  This output stream can be used as an input stream for another streaming algorithm, such as for diameter~\cite{shamosthesis78} and minimum enclosing rectangle~\cite{t83},
both of which rely on Shamos' rotating caliper method~\cite{shamosthesis78}. We apply Theorems~\ref{thm:det} and~\ref{thm:rand} to show Corollary~\ref{cor:diam}, detailed in Section~\ref{sec:proof}. 
\begin{corollary}\label{cor:diam}
Given a set $P$ of $n$ points in $\mathbb{R}^2$ where $|\ext{P}| = h$, there exists a deterministic W-stream algorithm to compute the diameter and minimum enclosing rectangles of $P$ in $O(\ceil{h/s}\log n)$ passes using $O(s)$ space and $O(n \log^2 n)$ time for every $s = \Omega(\log n)$. Given randomness, the pass complexity can be reduced to $O(h/s+ \log n)$ w.h.p. 
\end{corollary}

\smallskip

\noindent\textbf{Approximate Convex Hulls.} Given the hardness result shown in Theorem~\ref{thm:lower}, we know that one cannot have a constant-pass streaming algorithm that uses $o(h)$ space to compute the convex hull. In view of this, to have constant-pass $o(h)$-space streaming algorithms, one may consider  computing an approximate convex hulls. There are several results studying on how to efficiently find an approximate convex hull in the streaming model, based on a given error measurement.  The error criterion varies from the Euclidean distance~\cite{hs08}, and Hausdorff metric distance~\cite{Lopez00,Lopez05}, to the relative area error~\cite{rr15}. These algorithms use a single pass, $O(s)$ space, and can bound the given error measurement by a function of $s$. 

\smallskip

\noindent\textbf{Paper Organization.} In Section~\ref{sec:pre}, we present a new $O(n\log h)$-time RAM algorithm to compute the convex hull. Then, in Section~\ref{sec:algo}, we present a constant-pass streaming algorithm in the streaming model. In Section~\ref{sec:wrstr}, we present two W-stream algorithms, both of which use $O(s)$ space where $s$ can be as small as $O(\log n)$. We generalize the previous lower bound result in Section~\ref{sec:lower}, and prove a higher (but conditional) lower bound in Section~\ref{sec:clower}. We place the proofs of Corollary~\ref{cor:diam} in Section~\ref{sec:proof}.

%% file: pre.tex
\newcommand{\UH}[1]{U(#1)}

\section{Yet another $O(n \log h)$-time algorithm in the RAM model}\label{sec:pre}

Our streaming algorithm is based on a RAM algorithm, which we present in this section.  This RAM algorithm is a modification of  Kirkpatrick and Seidel's ultimate convex hull algorithm in the RAM model~\cite{ks86}. Chan and Chen's streaming algorithm~\cite{Chan07} is also based on Kirkpatrick and Seidel's algorithm, and thus the structure of these two streaming algorithms have some similarities. The changes are made so that our streaming algorithm does not have to rely on solving linear programs, thus reducing the computation cost compared to Chan and Chen's algorithm.

In what follows, we only discuss how to compute the upper hull because the lower hull can be computed analogously. Formally, computing the \emph{upper hull} $\UH{P}$ of a point set $P$ means outputting that part of the extreme points $v_1, v_2, \ldots, v_t \in \ext{P}$ in clockwise order so that $v_1$ is the leftmost point in $P$ and $v_t$ is the rightmost point in $P$, tie-breaking by picking the point with the largest $y$-coordinate, so that all points in $P$ lie below or on the line passing through $v_i, v_{i+1}$ for each $1 \le i < t$. Note that each of $v_1, v_2, \ldots, v_t$ has a unique $x$-coordinate, and each line that passes through $v_i$ and $v_{i+1}$ for $1 \le i < t$ has a finite slope. 

Roughly speaking, Kirkpatrick and Seidel's ultimate convex hull algorithm~\cite{ks86} evenly divides the point set into two subsets by a vertical line $\ell: x = \mu$, finds the hull edge in the upper hull that crosses $\ell$, and recurses on the two separated subsets. By appealing to the point-line duality, finding the crossing hull edge is equivalent to solving a linear program. Chan and Chen's streaming algorithm is adapted from this implementation of the ultimate convex hull algorithm. Their algorithm evenly divides the point set into $r+1$ subsets for $r \ge 1$ by $r$ vertical lines, finds the hull edges in the upper hull that cross these vertical lines, and recurses on the $r+1$ separated subsets. Finding these $r$ crossing hull edges is equivalent to solving $r$ linear programs, where the constraint sets for each are the same but the objective functions are different. 

\begin{table}[!h]
\centering
\begin{tabular}{|c|c|c|}
\hline 
 & Find $r$ hull edges, and recurse. & Find $r$ extreme points, and recurse. \\
\hline \hline 
$r = 1$ & Kirkpatrick and Seidel 1986~\cite{ks86} & Chan 1995~\cite{chanthesis95} \\
\hline 
any $r \ge 1$ & Chan and Chen 2007~\cite{Chan07}& This paper \\
\hline
\end{tabular}
\caption{Categorization of four $O(n \log h)$-time algorithms for convex hull. \label{tab:4algo}}
\end{table}

\vspace{-0.5cm}

In~\cite[Section 2]{chanthesis95}, Chan gives another version  of Kirkpatrick and Seidel's ultimate convex hull algorithm, that finds a suitable (possibly random) extreme point, divides the point set into two by $x$-coordinate, and recurses. The extreme point can be found by elementary techniques. Our streaming algorithm is adapted from the latter algorithm.  It finds $r$ suitable extreme points for $r \ge 1$, divides the point set into $r+1$ subsets by $x$-coordinate, and recurses on each subset. Though this generalization sounds straightforward, finding the $r$ suitable extreme points needs a different approach from that for finding a single suitable extreme point. We reduce finding these $r$ suitable extreme points to computing the upper hulls of $n/(r+1)$ small point sets. This reduction is the key observation of our RAM algorithm and is described in detail in the subsequent paragraphs. These four algorithms are categorized in Table~\ref{tab:4algo}.

Given $r$, our algorithm partitions $P$ arbitrarily into $G_1, G_2, \ldots, G_{n/(r+1)}$ so that each $G_j$ has size in $[1, r+1]$, and then computes the upper hull of each $G_j$. Let $Q$ be the union of the slopes of the hull edges in the upper hull of $G_1$, $G_2$, \ldots, $G_{n/(r+1)}$, which is a multiset. Let $\sigma_k$ be the slope of rank  $k|Q|/(r+1)$ in $Q$, for $k \in [1, r]$, in other words, $\sigma_k$ is the $k$th $(r+1)$-quantile in $Q$. To simplify the presentation, let $\sigma_0 = -\infty$ and $\sigma_{r+1} = \infty$. Let $s_k$ be the extreme point in $P$ that \emph{supports} slope $\sigma_k$, for each $k \in [0, r+1]$. That is, for every point $p \in P$ draw a line passing through $p$ with slope $\sigma_k$, and pick $s_k$ as the point whose line has the highest $y$-intercept. We define $s_0 = p_L$, the point with the smallest $x$-coordinate, and $s_{r+1} = p_R$, the point with the largest $x$-coordinate. If any $s_k$ has more than one candidates, pick the point that has the largest $y$-coordinate. Let $x(p)$ denotes the $x$-coordinate of point $p$, and let $\sigma(p, q)$ denote the slope of the line that passes through points $p$ and $q$.

We use these $s_1, s_2, \ldots, s_r$ as the $r$ \emph{suitable} extreme points with which to refine $P$ into $P_1, P_2, \ldots, P_{r+1}$ where we say the $s_i$ are suitable in that each $P_k$ has size bounded by $O(|P|/(r+1))$. Initially, set $P_k = \emptyset$ for all $k \in [1, r+1]$. The refinement applies the \emph{cascade-pruning} described in Lemma~\ref{lem:cascade-pruning} on $G_j$ for each $j \in [1, n/(r+1)]$, which uses the known pruning technique stated in Lemma~\ref{lem:prune} as a building block, and works as follows:
\begin{itemize}[leftmargin=1.8cm]
	\item[Step 1.] Compute $\UH{G_j}$, and obtain the extreme points $v_1, v_2, \ldots, v_t \in \UH{G_j}$ in clockwise order.
	\item[Step 2.] Set $P_k \leftarrow P_k \cup \{v_i : i \in [\alpha, \beta], x(s_{k-1}) < x(v_i) < x(s_k)\}$ for each $k \in [1, r+1]$, where $v_\alpha$ (resp. $v_\beta$) is the extreme point in $G_j$ that supports $\sigma_{k-1}$ (resp. $\sigma_{k}$).
\end{itemize}

The pruning in Step 2 is two-fold. For any $i < \alpha$, if $x(v_i) \le x(s_{k-1})$, then such a $v_i$ cannot be placed in $P_k$. Otherwise $x(v_i) > x(s_{k-1})$, and Case 2 of Lemma~\ref{lem:cascade-pruning} applies.  Again, such a $v_i$ cannot be placed in $P_k$. Similarly, $v_i$ for any $i > \beta$ cannot be placed in $P_k$ either. Finally, remove the points that lie below or on the line passing through $s_{k-1}, s_k$ from $P_k$ for each $k \in [1, r+1]$.

\begin{lemma}[Chan, \cite{chanthesis95}]\label{lem:prune}
Given a point set $P \subset \mathbb{R}^2$ and a slope $\sigma$, let $s$ be the extreme point in $P$ that supports $\sigma$. Then, for any pair of points $p, q \in P$ where $x(p) < x(q)$,
	\begin{itemize}[leftmargin=3cm]
		\item[Case 1.] If $\sigma(p, q) \le \sigma$ and $x(q) \le x(s)$, then $q \notin \UH{P}$.
		\item[Case 2.] If $\sigma(p, q) \ge \sigma$ and $x(p) \ge x(s)$, then $p \notin \UH{P}$.
	\end{itemize}
\end{lemma}

{
\renewcommand{\algorithmcfname}{RAM Algorithm}
\renewcommand{\thealgocf}{}
\begin{algorithm}[H]
	Let $G_1, G_2, \ldots, G_{n/(r+1)}$ be any partition of $P$ such that each $G_j$ has size in $[1, r+1]$\; 
	$Q \leftarrow \emptyset$\;
	\ForEach{$G_j$ in the partition}{
		Compute the upper hull $v_1, v_2, \ldots, v_t$ of $G_j$\;
		\For{$i = 1$ \KwTo $t-1$}{
			$\sigma \leftarrow $ the slope of the line passing through $v_i, v_{i+1}$\;
			$Q \leftarrow Q \cup \{\sigma\}$\;
		}
	}
	\For{$k = 1$ \KwTo $r$}{
		$\sigma_k \leftarrow \mbox{the } k|Q|/(r+1)$-th smallest slope in $Q$\;
		$s_k \leftarrow \mbox{the extreme point in $P$ that supports } \sigma_k$\;
	}
	$(s_0, \sigma_0, s_{r+1}, \sigma_{r+1}) \leftarrow (p_L, -\infty, p_R, \infty)$\;
	\For{$k = 1$ \KwTo $r+1$}{
		$P_k \leftarrow \emptyset$\;
		\ForEach{$G_j$ in the partition}{
			Compute the upper hull $v_1, v_2, \ldots, v_t$ of $G_j$\;
			Find the extreme point $v_{\alpha}$ (resp. $v_\beta$) in $G_j$ that supports $\sigma_{k-1}$ (resp. $\sigma_k$)\;
			$P_k \leftarrow P_k \cup \{v_{\alpha}, v_{\alpha+1}, \ldots, v_{\beta}\}$\;
		}
		Remove the points that lie below or on the line passing through $s_{k-1}$, $s_{k}$ from $P_k$\;
		\If{$P_k \ne \emptyset$}{
			Recurse on $P_k \cup \{s_{k-1}, s_k\}$\;
		}
	}
\caption*{Compute the upper hull $\UH{P}$ of $P$. \label{fig:RAMalgo}}
\end{algorithm}
}
\addtocounter{algocf}{1}

\begin{lemma}[Cascade-pruning]\label{lem:cascade-pruning}
Given a point set $P \subset \mathbb{R}^2$ and a slope $\sigma$, let $s$ be the extreme point in $P$ that supports $\sigma$. Then, for any $G \subseteq P$ whose $\UH{G} = \{v_1, v_2, \ldots, v_t\}$, $x(v_1) < x(v_2) < \cdots < x(v_t)$, and where $\delta \in [1, t]$ is such that $v_\delta$ is the extreme point in $G$ that supports $\sigma$, we have:
	\begin{itemize}[leftmargin=2cm]
		\item[Case 1.] If $x(v_i) \le x(s)$ for some $i \in [\delta+1, t]$, then $v_{\delta+1}, \ldots, v_{i} \notin \UH{P}$.
		\item[Case 2.] If $x(v_i) \ge x(s)$ for some $i \in [1, \delta-1]$, then $v_{i}, \ldots, v_{\delta-1} \notin \UH{P}$.
	\end{itemize}
\end{lemma}
\begin{proof}
Observe that $\sigma(v_{j}, v_{j+1}) \ge \sigma$ for all $j \in [1, \delta-1]$ and $\sigma(v_{j-1}, v_{j}) \le \sigma$ for all $j \in [\delta+1, t]$ because $v_1, v_2, \ldots, v_t$ are extreme points in $\UH{G}$ in clockwise order and $v_\delta$ is the extreme point in $G$ that supports $\sigma$. Since there is an  $i \in [\delta+1, t]$ such that  $x(v_i) \le x(s)$, we have $x(v_j) \le x(s)$ for each $j \in [\delta+1, i]$. The above are exactly the conditions of Case 1 in Lemma~\ref{lem:prune} for all point pairs $(v_{j-1}, v_{j})$ whose $j \in [\delta+1, i]$. Thus, $v_{j} \notin \UH{P}$ for all $j \in [\delta+1, i]$. The other case can be proved analogously.
\end{proof}

We get the exact bound for each $P_k$ in Lemma~\ref{lem:balanced}, noting that $|P_k| \le \frac{3}{4}|P|$ for $r=1$. 

\begin{lemma}\label{lem:balanced}
$
	|P_k| \le (\frac{2}{r+1}-\frac{1}{(r+1)^2}) |P| \le 2|P|/(r+1) \mbox{ for each } k \in [1, r+1]. 
$
\end{lemma}

\begin{proof}
To ensure that, for every $k \in [1, r+1]$, $P_k$ is a small fraction of $P$, we use the cascade-pruning procedure described in Lemma~\ref{lem:cascade-pruning}. Let $\{v_{1}, v_{2}, \ldots, v_{t}\}$ be $\UH{G_j}$ for some $j \in [n/(r+1)]$ where $x(v_1) < x(v_2) < \cdots < x(v_t)$. Let $v_{\alpha_j}$ (resp. $v_{\beta_j}$) be the extreme point in $G_j$ that supports $\sigma_{k-1}$ (resp. $\sigma_{k}$).

Let $n_j$ be the number of points in $P_k \cap G_j$. Recall that $P_k$ does not contain any $v_i$ for any $i \notin [\alpha_j, \beta_j]$, and hence $n_j \le \beta_j - \alpha_j+1$. Observe that point pair $(v_{i}, v_{i+1})$ has slope in the open interval $(\sigma_{k-1}, \sigma_{k})$ for each $i \in [\alpha_j, \beta_j-1]$. Since $\sigma_{k-1}$ (resp. $\sigma_{k}$) is the $(k-1)|Q|/(r+1)$-th largest slope (resp. the $k|Q|/(r+1)$-th largest slope) in $Q$, $Q$ has at most $|Q|/(r+1)$ slopes in the open interval $(\sigma_{k-1}, \sigma_{k})$. This yields that
$$
	\sum_{j = 1}^{n/(r+1)} n_j-1 \le \frac{|Q|}{r+1} \Rightarrow \sum_{j = 1}^{n/(r+1)} n_j \le \frac{|Q|}{r+1} + \frac{n}{r+1} \le \frac{r|P|}{(r+1)^2} + \frac{|P|}{r+1}
$$
The last inequality holds because $|Q| \le r|P|/(r+1)$, and it establishes that the number of points from all $G_j$'s that comprise $P_k$ for each $k \in [1, r+1]$ is at most $2|P|/(r+1)$. 
\end{proof}

For each $k \in [1, r+1]$, if $P_k \neq \emptyset$, then our algorithm recurses on $P_k \cup \{s_{k-1}, s_k\}$. This ensures that every subproblem has an input that contains some intermediate extreme point(s), i.e. not the leftmost and rightmost extreme points, and any two subproblems where one is not an ancestor or a descendant of the other have an empty intersection in their intermediate extreme point set. As a result, 
\begin{lemma}\label{lem:leaves}
Our algorithm has $O(h)$ leaf subproblems. 
\end{lemma}
We need Lemma~\ref{lem:leaves} to analyze the running time. 

\subsection*{Running Time}

Here we analyze the running time of the RAM algorithm for the case of $r = O(1)$ and defer the discussion for the case of $r = \omega(1)$ until the section on streaming algorithms. Let $T_C$ be the recursive computation tree of the RAM algorithm. The root of $T_C$ represents the initial problem of the recursive computation. Every node in $T_C$ has at most $r+1$ child nodes, each of which represents a recursive subproblem. 

For a computation node with the input point set $P$ whose $|P| < r$, we use any $O(|P|\log r)$-time algorithm to compute the convex hull. Otherwise, we need to compute $|P|/(r+1)$ convex hulls of point sets of size at most $r+1$, which runs in $O(|P|\log r)$ time (Lines 1-9). In addition, the quantile selection in $Q$ has the running time $O(|Q|\log r) = O(|P|\log r)$ (Line 11). The $r$ suitable extreme points can be found in $O(|P|\log r)$ time by Lemma~\ref{lem:getextp} (Line 12). The pruning procedure can be done in $O(|P|\log r)$ time by a simple merge (Lines 15-26). Hence, each computation node needs $O(|P|\log r)$ time. 

Since each child subproblem has an input set $P_k \cup \{s_{k-1}, s_k\}$ of size at most $2|P|/(r+1)+2$ (Lemma~\ref{lem:balanced}), the running time of child subproblem is an $(2/(r+1))$-fraction of its parent subproblem. Hence, $T_C$ is an \emph{$(2/(r+1))$-fading} computation tree where Edelsbrunner and Shi~\cite{es91} define a recursive computation tree to be $\alpha$-fading for some $\alpha < 1$ if the running time of a child subproblem is an $\alpha$-fraction of its parent. In~\cite{chanthesis95}, Chan extends Edelsbrunner and Shi's results and obtains that, if an $\alpha$-fading recursive computation tree has $L$ leaf nodes and the total running time of the nodes on each level is at most $F$, then the recursive computation tree has total running time $O(F \log L)$. Our algorithm has $O(h)$ leave nodes (Lemma~\ref{lem:leaves}) and $O(|P|\log r)$ time for the computation nodes on each level because two subproblems on the same level have their inputs only intersected at one of their extreme points. We get:

\begin{theorem}\label{thm:ram}
	The RAM algorithm runs in $O(n \log h \log r)$ time, and for $r = O(1)$ it is an $O(n \log h)$-time algorithm.
\end{theorem}

%% file: algo.tex
\section{A Simpler and Faster Streaming Algorithm}\label{sec:algo}

In this section, we show how to adapt our RAM algorithm to the streaming model. Our streaming algorithm is the same as our RAM algorithm, but we execute the subproblems on $T_C$ in BFS order. That is, starting from the root of $T_C$, all subproblems on $T_C$ of the same level are solved together in a round, then their invoked subproblems are solved together in the next round, and so on. We will see in a moment that our algorithm needs to scan the input $O(1)$ times for each round. Therefore, to have an $O(1)$-pass streaming algorithm, our approach requires $r = n^\delta$ for some positive constant $\delta < 1$. By setting $r = n^\delta$, we have:

\begin{lemma}\label{lem:internal}
By setting the parameter $r$ to be $n^\delta$ for any constant $\delta \in (0, 1)$, the recursive computation tree $T_C$ has $O\left(\delta^{-1} h \right)$ nodes.
\end{lemma} 
\begin{proof}
This lemma holds because $T_C$ has depth $O(\log_r n) = O(\delta^{-1})$ by Lemma~\ref{lem:balanced} and $T_C$ has $O(h)$ leaf nodes by Lemma~\ref{lem:leaves}. 
\end{proof}

We assign a unique identifier $z \in [1, |T_C|]$ to each of $|T_C| = O(\delta^{-1}h)$ subproblems. Let $S_z$ be the subproblem on $T_C$ whose identifier is $z$. For each $z \in [1, |T_C|]$, $S_z$ has input point set $P_z$. $P_z$ is a subsequence of $P$ and is given to $S_z$ as an input stream of $|P_z|$ points. Our algorithm will generate $P_z$ more than once for $S_z$ to access, for all $z \in [1, |T_C|]$. The data structures used in $S_z$ also are suffixed with $z$. To compute $S_z$, naively we need $O(|P_z|)$ space. We will see in a moment that given $P_z$, how to solve $S_z$ using $O(r \log r |P_z|)$ space in $O(r \log |P_z|+ |P_z|\log r)$ time. We will also see how to generate the input for all the subproblems on $T_C$ of depth $d > 0$ in $O(1)$ passes. 
We now establish all these claims, after which we will be ready to prove Theorem~\ref{thm:read-only}.
We decompose $S_z$ into the following three subtasks and describe the algorithms for the subtasks in the subsequent subsections.

\begin{enumerate}
	\item Given $P_z$, obtain the $r$ quantile slopes $\sigma_1, \sigma_2, \ldots, \sigma_r$.
	\item Given $P_z$ and $\sigma_1, \sigma_2, \ldots, \sigma_r$, obtain the $r$ suitable extreme points $s_1, s_2, \ldots, s_r$.
	\item After the ancestor subproblems of $S_z$ (excluding $S_z$) are all solved, generate $P_z$.
\end{enumerate}

\subsection{Obtaining the $r$ quantile slopes}

To find the $r$ quantile slopes for $S_z$ (Lines 1-11 in the RAM algorithm) using small space, we use a Greenwald and Khanna~\cite{gk01} quantile summary structure, abbreviated as $QS_z$. This summary is a data structure that supports two operations: insert a slope ($QS_z$.insert($\sigma$)) and query for (an estimate of) the $t$-th smallest slope ($QS_z$.query($t$)) in $Q_z$. Given access to $QS_z$, we do not have to store the entire $P_z$ to obtain the $r$ quantile slopes. Instead, we invoke $QS_z$.insert($\sigma$) for each slope $\sigma \in Q_z$. After updating all slopes in $Q_z$, we obtain an estimate of the $(r+1)$-quantile of $Q_z$ by invoking $QS_z$.query($k|Q_z|/(r+1)$) for all $k \in [1, r]$. The detailed implementation of the above adaption to the streaming model is given in Algorithm~\ref{algo:slope}. 

\begin{algorithm}[H]
 Initialize $QS_z$\;
 $B_z \leftarrow \emptyset$\; 
 $q_z \leftarrow 0$\tcc*{$q_z$ counts $|Q_z|$}
 \ForEach{$p$ in $P_z$}{
   $B_z \leftarrow B_z \cup \{p\}$\;
	 \If{ $|B_z|$ equals $r+1$ \textbf{or} $p$ is the last point in $P_z$}{
	 	Compute the upper hull $v_1, v_2, \ldots, v_t$ of $B_z$\;
		\ForEach{$i = 1$ \KwTo $t-1$}{
			$QS_z$.insert($\sigma(v_i, v_{i+1})$)\;
			$q_z \leftarrow q_z+1$\;
		}
		$B_z \leftarrow \emptyset$\;
	 }
 }
 \ForEach{$k = 1$ \KwTo $r$}{
 	$\hat{\sigma}_k \leftarrow QS_z$.query($kq_z/(r+1)$)\;
 }
\caption{Compute the $r$ approximate quantile slopes for the subproblem $S_z$. \label{algo:slope}}
\end{algorithm}

$QS_z$.query($k|Q_z|/(r+1)$) returns an estimate $\hat{\sigma}_k$ that has an additive error $c|Q_z|$ in the rank, where $c$ is a parameter to be determined. We set $c = \epsilon/(r+1)$ for some constant $\epsilon > 0$ so that the additive error cannot increase the depth of $T_C$ by more than a constant factor. Precisely, because the obtained $\hat{\sigma}_k$ has the rank in the range 
$$
	[(k-\epsilon)|Q_z|/(r+1), (k+\epsilon)|Q_z|/(r+1)]
$$ 
for each $k \in [1, r]$, we need to replace Lemma~\ref{lem:balanced} with Corollary~\ref{cor:balanced}. Such a replacement increases the depth of $T_C$ from $O(\log_{r} n) = O(\delta^{-1})$ to $O(\log_{r/(1+\epsilon)} n) = O(\delta^{-1})+o(1)$.

\begin{corollary}\label{cor:balanced}
$|P_k| \le (\frac{2+2\epsilon}{r+1}-\frac{1}{(r+1)^2})|P| \le 2(1+\epsilon)|P|/(r+1) \mbox{ for each } k \in [1, r+1].$
\end{corollary}

The summary $QS_z$ needs $O\left(\frac{1}{c} \log (c|Q_z|)\right)$ space, and therefore the space usage for each subproblem is $O((r/\epsilon) \log((\epsilon/r)|Q_z|))$. In~\cite{yccs16}, it shows that Greenwald and Khanna's quantile summary needs $O(\log |Q_z|)$ time for an update and $O(\log r + \log\log(|Q_z|/r))$ for a query. Because $S_z$ conducts $O(r)$ updates and $O(r)$ queries, we get:

\begin{lemma}\label{lem:getslope}
Given $P_z$, there exists a streaming algorithm that can obtain the $r$ approximate quantile slopes in $Q_z$ to within any constant factor using $O(r \log(|P_z|+r))$ time and $O(r \log(|P_z|/r))$ space.
\end{lemma}

\subsection{Obtaining the $r$ suitable extreme points}

To find the $r$ suitable extreme points in $P_z$ (Line 12 in the RAM algorithm), a naive implementation, which would update the supporting points of $\hat{\sigma}_k$ for all $k \in [1, r]$ once for each point $p \in P_z$, needs $O(r|P_z|)$ running time. To reduce the running time to the claimed time complexity $O(r \log |P_z| + |P_z| \log r)$, we need the following observation.

\begin{observation}\label{obs:pred-succ}
For any non-singleton set $G$ whose extreme points in the upper hull $\UH{G}$ from left to right are $v_1, v_2, \ldots, v_t$, the point in $G$ that supports a given slope $\sigma$ is
$$
	s = \left\{
	\begin{array}{ll}
		v_1 & \mbox{ if } \sigma > \sigma(v_1, v_2) \\
		v_t & \mbox{ if } \sigma < \sigma(v_{t-1}, v_t) \\
		v_i & \mbox{ if } \sigma(v_{i-1}, v_i) \ge \sigma \ge \sigma(v_i, v_{i+1}) \mbox{ for some } i \in [2, t-2]
	\end{array}
	\right.
$$
\end{observation}

To find the extreme points in $P_z$ that supports $\hat{\sigma}_k$ for all $k \in [1, r]$, we compute the extreme points $v_1, v_2, \ldots, v_t$ in $P_z$ from left to right, generate a (sorted) list $\ell_A$ of slopes $\sigma(v_1, v_2), \sigma(v_2, v_3), \ldots, \sigma(v_{t-1}, v_t)$, and merge $\ell_A$ with another (sorted) list $\ell_B$ of the approximate $(r+1)$-quantile slopes $\hat{\sigma}_1, \hat{\sigma}_2, \ldots, \hat{\sigma}_r$. By Observation~\ref{obs:pred-succ}, the point $\hat{s}_k$ in $P_z$ that supports $\hat{\sigma}_k$ for each $k \in [1, r]$ can be easily determined by the its predecessor and successor in $\ell_A$. Scanning the merged list suffices to get $\hat{s}_1, \hat{s}_2, \ldots, \hat{s}_k$. Though the above reduces the time complexity to $O(r + |P_z| \log |P_z|)$, the space complexity $O(|P_z|)$ is much higher than the claimed space complexity $O(r \log r|P_z|)$ for $r \ll |P_z|$. To remedy, again, we reduce this problem to computing the upper hulls of $|P_z|/(r+1)$ smaller point sets. First, we partition $P_z$ arbitrarily into $G_1, G_2, \ldots, G_{|P_z|/(r+1)}$ so that each group $G_i$ has size $|G_i| \in [1, r+1]$ points. Then, for each $G_i$ we apply the above accordingly, detailed in Algorithm~\ref{algo:support}. We get:

\begin{lemma}\label{lem:getextp}
Given $P_z$ and sorted $\sigma_1, \sigma_2, \ldots, \sigma_r$, there exists a streaming algorithm that can obtain the extreme points in $P_z$ that support $\sigma_i$ for all $i \in [1, r]$ using $O(r+|P_z|\log r)$ time and $O(r)$ space.
\end{lemma}

\begin{algorithm}[H]
 $B_z \leftarrow \emptyset$\; 
 $\hat{s}_k \leftarrow (0, -\infty)$ for each $k \in [1, r]$\;
 \ForEach{$p$ in $P_z$}{
   $B_z \leftarrow B_z \cup \{p\}$\;
	 \If{ $|B_z|$ equals $r+1$ \textbf{or} $p$ is the last point in $P_z$}{
	 	Compute $\UH{B_z}$ and obtain its extreme points from left to right, $v_1, v_2, \ldots, v_t$\;
		$\ell_A \leftarrow \sigma(v_1, v_2), \sigma(v_2, v_3), \ldots, \sigma(v_{t-1}, v_t)$\;
		$\ell_B \leftarrow \hat{\sigma}_1, \hat{\sigma}_2, \ldots, \hat{\sigma}_r$\;
		Merge $\ell_A$ and $\ell_B$\;
		\ForEach{$k = 1$ \KwTo $r$}{
			Find the predecessor and successor of $\hat{\sigma}_k$ in $\ell_A$ by scanning the merged list\;
			By which and Observation~\ref{obs:pred-succ}, obtain the point $b_k$ in $B_z$ that supports $\hat{\sigma}_k$\;
			$\hat{s}_k \leftarrow$ the point in $\{\hat{s}_k, b_k\}$ that supports $\hat{\sigma}_k$\;
		}	
	 }
 }
\caption{Compute the $r$ suitable extreme points for the subproblem $S_z$. \label{algo:support}}
\end{algorithm}

\subsection{Generating the input point set $P_z$ for each subproblem $S_z$}

Recall that we execute the subproblems in $T_C$ in BFS order. Upon executing the subproblems of depth $d$ for any $d > 0$, all the subproblems of depth less than $d$ are done and the associated $r$ quantile slopes and $r$ suitable extreme points are memoized in memory. For $d = 0$, we need to generate the input for the initial problem $S_o$. Because its input point set is exactly $P$, scanning over $P$ suffices. 

Given the associated $r$ quantile slopes and $r$ suitable extreme points for all the subproblems of depth less than $d$, to generate the input point sets for all the subproblems of depth $d$, we can directly execute Lines 15-26 in the RAM algorithm for all the subproblems of depth less than $ d$ and ignore Lines 1-14 because the intermediate values, the quantile slopes and suitable extreme points, are already computed and kept in memory. Initially, we allocate a buffer $B_z$ of size $r+1$ for each subproblem $S_z$ of depth less than $d$ so as to temporarily store the incoming input points, i.e. points in $P_z$. Then, we scan $P$ on the input tape once and for each input point $p$ in $P$, we place $p$ in the buffer $B_o$ of $S_o$. Once any buffer $B_z$ gets full or the input terminates, we let $B_z$ be some $G_i$, a part in the partition of $P_z$, and apply the pruning procedure stated in Lines 15-26 in the RAM algorithm. Those points that survive the pruning are flushed, one by one, into the buffers of $S_z$'s child subproblems. We apply the above iteratively until we reach the end of the input tape. The space usage counted on each $S_z$ is $O(|B_z|) = O(r)$ and the overall running time to generate the input point set for all the subproblems of depth $d > 0$ is $O(d n \log r)$ because all the subproblems of each depth $i \in [1, d-1]$ computes the upper hull of points sets, disjoint subsets of $P$. Hence, we get:

\begin{lemma}\label{lem:geninput}
There exists a streaming algorithm that can generate the input for all the subproblems on $T_C$ of depth $d$ for each $d \in [0, \depth{T_C}]$ using $O(1)$ passes, $O(hr)$ space, and $O(d n \log r)$ time.  
\end{lemma}

\begin{proof}[Proof of Theorem~\ref{thm:read-only}]
For $r = n^\delta$, $T_C$ has $O(\delta^{-1}h)$ nodes and depth $O(\delta^{-1})$ by Lemmas~\ref{lem:internal} and~\ref{lem:balanced}. Hence, the space complexity of our streaming algorithm is the sum of $O(\delta^{-1}h)$ times the space complexity in Lemmas~\ref{lem:getslope} and~\ref{lem:getextp}, and $O(\delta^{-1})$ times the space complexity in Lemma~\ref{lem:geninput}. The overall space complexity is $O(\delta^{-1} h n^\delta\log n)$. One can obtain the space bound $O(\min\{\delta^{-1} hn^\delta\log n, n\})$ by checking whether $\hbar n^\delta \log n > n$ before proceeding to the subproblems on the next depth, where $\hbar$ is the number of subproblems executed so far and thus $\hbar = O(\delta^{-1}h)$. If so, we compute the convex hull by a RAM algorithm. 

Analogously, we have that the pass (resp. time) complexity of our streaming algorithm is $O(\delta^{-1})$ (resp. $O(\delta^{-2} n \log n)$)
\end{proof}

%% file: wrstr.tex
\section{A W-Stream Algorithm Of Nearly-Optimal Pass-Space Tradeoff}\label{sec:wrstr}

Demetrescu et al.~\cite{cbga10} establish a general scheme to convert PRAM algorithms to W-stream algorithms. Theorem~\ref{thm:parallel} is an implication of their main result. 

\begin{theorem}[Demetrescu et al.~\cite{cbga10}] \label{thm:parallel}
If there exists a PRAM algorithm that uses $m$ processors to compute the convex hull of $n$ given points in $t$ rounds, then there exists an $O(s)$-space $O(mt/s)$-pass W-stream algorithm to compute the convex hull.
\end{theorem}

There is a long line of research that studies how to compute the convex hull of $n$ given points efficiently in parallel~\cite{chow80,ag86,jag86,akl84,gg91,gs97}. In particular, Akl's PRAM algorithm~\cite{akl84} uses $O(n^\epsilon)$ processors and runs in $O(n^{1-\epsilon}\log h)$ time for any $\epsilon \in (0, 1)$. Converting Akl's PRAM algorithm to a W-stream algorithm by Theorem~\ref{thm:parallel}, we have:

\begin{corollary}\label{cor:auto}
There exists an $O((n/s)\log h)$-pass W-stream algorithm that can compute the convex hull of $n$ given points using $O(s)$ space. 
\end{corollary}

The optimal work, i.e., the total number of primitive operations that the processors perform, for any parallel algorithm in the algebraic decision tree model\footnote{Roughly speaking, algorithms are decision trees in which each computation node is able to test the sign of the evaluation of a constant-degree polynomial.} 
to compute the convex hull is $O(n \log h)$~\cite{ks86,gs97}. Therefore the W-stream algorithm stated in Corollary~\ref{cor:auto} is already the best possible among those W-stream algorithms that are converted from a PRAM algorithm in the algebraic decision tree model by Theorem~\ref{thm:parallel}. However, in this Section, we will show that such a tradeoff between pass complexity and space usage is suboptimal by devising a W-stream algorithm that has a better pass-space tradeoff. Together with the results shown in Section~\ref{sec:lower}, we have that the pass-space tradeoff of our W-stream algorithm is nearly optimal.

\subsection{Deterministic W-stream Algorithm}

Our deterministic W-stream algorithm is the same as our streaming algorithm, except for the following differences:

\begin{itemize}
\item We set $r = 1$ (rather than $r = n^\delta$) for our deterministic W-stream algorithm. Thus, by Corollary~\ref{cor:balanced} $\depth{T_C}$ increases from $O(\delta^{-1})$ to $O(\log n)$, but the space usage of subproblem $S_z$ decreases from $O(n^\delta \log n)$ to $O(\log n)$ for each $z \in [1, |T_C|]$. Moreover, if the extreme point in the input $P$ that supports the approximate median slope is the leftmost point $p_L$ or the rightmost point $p_R$, i.e. the degenerate case, we replace it with the extreme point that supports $\sigma(p_L, p_R)$. In this way, each subproblem on $T_C$ has a unique extreme point and therefore the number of subproblems on $T_C$ is $O(h)$.

\item Our streaming algorithm executes the subproblems on $T_C$ in BFS order, that is, all subproblems of depth $d$ are executed in a round for each $d \in [0, \depth{T_C}]$. In contrast, our deterministic W-stream algorithm refines a single round into subrounds, in each of which it takes care of $O(s/\log n)$ subproblems, so as to bound the working space by $O(s)$.

\item Note that algorithms in the W-stream model are capable of modifying the input tape. Formally, while scanning the input tape in the $i$-th pass, algorithms can write something on a write-only output stream; in the $(i+1)$-th pass, the input tape read by algorithms is the output tape written in the $i$-th pass. Hence, our deterministic W-stream algorithm is able to assign an attribute to each point $p \in P$ to indicate that $p$ is an input of a certain subproblem. Moreover, our deterministic W-stream algorithm can write down the parameters for every subproblem on the output tape. In each subround, our deterministic W-stream algorithm needs to scan the input twice. The first pass is used to load the parameters of subproblems to be solved in the current subround. The second pass is used to scan the input tape and process the points that are the input points for the subproblems to be solved in the current subround.  
\end{itemize}

We are ready to prove Theorem~\ref{thm:det}.  

\begin{proof}[Proof of Theorem~\ref{thm:det}]
Suppose there are $h_d$ subproblems of depth $d$ on $T_C$ for each $d \in [0, \depth{T_C}]$, then our deterministic W-stream algorithm has to execute 
$$
	\sum_{d \in [0, \depth{T_C}]} \left\lceil \frac{h_d}{\lfloor s/\Theta(\log n) \rfloor} \right\rceil = O\left(\lceil h/s \rceil \log n\right)
$$
subrounds for any $s = \Omega(\log n)$. Because our deterministic W-stream algorithm scans the input tape twice for each subround, the pass complexity is $O(\lceil h/s \rceil \log n)$. 

As shown in Section~\ref{sec:algo}, subproblem $S_z$ needs $O(|P_z| \log |P_z|)$ running time. Since the input of subproblems of depth $d$ on $T_C$ are disjoint subsets of $P$, for each $d \in [0, |T_C|]$. We get that the time complexity is $O(n \log^2 n)$. 
\end{proof}

\subsection{Randomized W-stream Algorithm}

Observe that for $r = 1$, finding the $r$ approximate quantile slopes in $Q_z$ is exactly finding the approximate median slope in $Q_z$. Our algorithms mentioned previously all use Greenwald and Khanna quantile summary structure, which needs $O(\log n)$ space. In our randomized W-stream algorithm, we replace the Greenwald and Khanna quantile summary with a random slope in $Q_z$, thereby reducing the space usage to $O(1)$. As noted by Bhattacharya and Sen~\cite{bs97}, such a replacement cannot increase the depth of $T_C$ by more than a constant factor w.h.p. Consequently, we get Theorem~\ref{thm:rand}.

\begin{proof}[Proof of Theorem~\ref{thm:rand}]
Similar to the arguments used in the proof of Theorem~\ref{thm:det}, the pass complexity of our randomized W-stream algorithm is
$$
	\sum_{d \in [0, \depth{T_C}]} \left\lceil \frac{h_d}{\lfloor s/\Theta(1) \rfloor} \right\rceil = O\left(h/s+\log n\right)
$$
for any $s = \Omega(1)$ w.h.p. and the time complexity is $O(n \log^2 n)$ w.h.p.
\end{proof}

%% file: lower.tex
\newcommand{\tcite}{Kalyanasundaram and Schintger~\cite{ks92}\xspace}

\section{Unconditional Lower Bound}\label{sec:lower}

In this section, we will show that any streaming (or W-stream) algorithm that can compute the convex hull with success rate $> 2/3$ using $O(s)$ space requires $\Omega(\ceil{h/s})$ passes (i.e. Theorem~\ref{thm:lower}). This establishes the near-optimality of our proposed algorithms. We note here that the lower bound holds even if the output is the quantity $|\ext{P}|$, rather than the set $\ext{P}$.

\vspace{-0.5cm}
\input{fig2} 

We construct a point set $U$ so that it is hard to compute the convex hull of point set $P = Q \cup \{(1, 0), (-1, 0)\}$
for all $Q \subseteq U$, as illustrated in Figure~\ref{fig:hard}. Let $C_1, C_2$ be concentric half circles. The radius of $C_1$ equals 1 and that of $C_2$ is any value in $(k,1)$ for some $k$ to be determined later. Let $a_0, a_1, \ldots, a_{n+1}$ be points distributed evenly on $C_1$ so that $a_0 = (1, 0)$ and $a_{n+1} = (-1, 0)$. Define $b_0, b_1, \ldots, b_{n+1}$ on $C_2$ similarly. Let $k$ be the distance between the origin $O$ and the line $\overleftrightarrow{a_i a_{i+2}}$ for any $i \in [0, n-1]$. Let $U$ be the set $\{a_i : i \in [1, n]\} \cup \{b_i : i \in [1, n]\}$. 

Before proceeding to the hardness proof, observe the following geometric property of points in $U$. 

\begin{lemma}\label{lem:geo}
For every $Q \subseteq U$, let $R = \ext{Q \cup \{(1, 0), (-1, 0)\}}$. We have that
	\begin{enumerate}[leftmargin=1.5cm,label={(\arabic*)}]
		\item If $a_i \in Q$, then $a_i \in R$. 
		\item If $b_i \in Q$, then $b_i \in R$ iff $a_i \notin Q$.
	\end{enumerate}
\end{lemma}
\begin{proof}
(1) Since $a_i$ is on $C_1$, $a_i$ cannot be expressed as a convex combination of any other points in $C_1 \setminus \{a_i\}$. That implies $a_i$ is an extreme point of $Q$ as long as $a_i \in Q$, and the same argument holds for every $i$.

(2 $\Rightarrow$) If $a_i \in Q$, then $b_i \in \mathrm{\Delta}a_0a_ia_{n+1}$ and thus $b_i \notin R$.

(2 $\Leftarrow$) If $a_i \notin Q$, then $b_i \in R$. To see why, we draw a tangent line $L_{b_i}$ of $C_2$ passing through $b_i$ as illustrated in Figure~\ref{fig:tangentLineFig}. Since $r > k$, all the points in $U \setminus \{a_i\}$ are strictly on one side of $L_{b_i}$, implying that $b_i$ cannot be expressed as a convex combination of any other points in $R \setminus \{a_i\}$ for all $R$. Thus, $b_i$ is an extreme point if $a_i \notin Q$.
\end{proof}

Lemma~\ref{lem:geo} implies the fact that, for every $Q \subseteq U$, 
$$
	|\ext{Q \cup \{(1, 0), (-1, 0)\}}| = |Q|+2 
$$
if and only if $a_i$ and $b_i$ are not both contained in $Q$ for each $i$. Given this fact, we are ready to perform a reduction from the \emph{set disjointness} problem (a two-party communication game) to computing the convex hull in the streaming (and W-stream) model. Set disjointness is defined as follows:

\begin{itemize}[leftmargin=1.5cm]
	\item[Given:] Alice has a private $(\alpha n)$-size subset $A$ of $[n]$, and Bob has another private $(\alpha n)$-size subset $B$ of $[n]$ for some constant $\alpha < 1/2$. 
	\item[Goal:\;\,] Answer whether $A$ and $B$ have an non-empty intersection. 
\end{itemize}

Kalyanasundaram and Schintger~\cite{ks92} show that

\begin{theorem}[\tcite]\label{lem:disj}
No matter which 2-way, multi-round protocol Alice and Bob use, they must communicate $\Omega(n)$ bits to answer the set disjointness problem with constant success rate greater than $2/3$.
\end{theorem}

We are ready to proceed to the proof of Theorem~\ref{thm:lower}.
\begin{proof}[Proof of Theorem \ref{thm:lower}]
We claim that, if there exists a streaming (or W-stream) algorithm that can compute the convex hull using $s$ bits in $p$ passes with constant probability greater than $2/3$, then the set disjointness problem can be answered using $(ps)$ bits with constant probability greater than $2/3$. Hence, $ps = \Omega(n)$ or $p = \Omega(\ceil{n/s})$, noting that $p$ is an integer and cannot be a sub-constant. 

To prove the claim, for every $A, B \in [n]$ we map them to $Q_A, Q_B \subseteq U$, so that $Q_A = \{a_i : i \in A\}$ and $Q_B = \{b_i : i \in B\}$. Then, we use a tape to store the points in $Q_A \cup Q_B \cup \{(1, 0), (-1, 0)\}$, where $Q_A$ occupies the former half of the tape and the rest of points occupies the latter half of the tape. If there exists an algorithm that can compute the convex hull of $R = Q_A \cup Q_B \cup \{(1, 0), (-1, 0)\}$, it must know what $|\ext{R}|$ is. Given the above fact, we have that $A \cap B = \emptyset$ iff $|\ext{R}| = |Q_A| + |Q_B| +2$. This gives us the ability to solve the set disjointness problem. 

Then, we generate four sets of the above arrangement. Let the input of these sets be (1) $A$ and $B$, (2) $\overline{A}$ and $B$, (3) $A$ and $\overline{B}$, and (4) $\overline{A}$ and $\overline{B}$. Instead of distributing the generated points among a half circle, we distribute the generated points among an eighth circle for each set. Then, we concatenate these eighth circles so that they evenly partition a half circle. No matter what $A$ and $B$ are, the total number of extreme points in these sets is $3n$. Since any algorithm that computes convex hull needs to output the extreme points in order, as defined, if we observe the outputted extreme points in the first eighth circle, it suffices to answer the above set disjointness problem, while retaining the number of extreme points to be a fixed value $3n$, given $n$. 

Hence, we are able to reduce the set disjointness problem of domain set $[h]$ to computing the convex hull of $3h$ extreme points, one can place $n-3h$ dummy points to the locations that are very close to origin, and thus all dummy points are interior points no matter what $Q_A$ and $Q_B$ are. This establishes Theorem~\ref{thm:lower}.
\end{proof}

%% file: fig2.tex
\begin{figure}[ht]
\centering
\begin{subfigure}[t]{0.4\textwidth}
\centering
\psset{unit=0.6, linewidth=0.015,dotsize=1.7pt}
\begin{pspicture}(-6,-1)(6,6)
\psline(-5,0)(5,0)
\psarc(0,0){5}{0}{180}
\psline(-4,3)(0,0) \psline(0,0)(4,3)
\psline(-4.7434,1.5811)(0,0) \psline(4.7434,1.5811)(0,0)
\psline(-5,0)(-4,3) \psline(5,0)(4,3)
\rput(0,-0.5){$O$} \rput(5.5,0){$a_0$} \rput(-5.8,0){$a_{n+1}$}
\rput(5.2,1.5811){$a_1$} \rput(-5.2,1.5811){$a_n$} \rput(4.5,3){$a_2$} \rput(-4.7,3){$a_{n-1}$}
\psarc(0,0){4.87}{0}{180}
\psdots[dotstyle=*,linecolor=red](-5,0)(-4.7434,1.5811)(-4,3)(4,3)(4.7434,1.5811)(5,0)
\psdots[dotstyle=*,linecolor=red](-4.87,0)(-4.61,1.54)(-3.89,2.922)(3.89,2.922)(4.61,1.54)(4.87,0)
\rput(4.5,0.3){$b_0$} \rput(-4.2,0.3){$b_{n+1}$} \rput(3.5,2.922){$b_2$} \rput(-3.1,2.922){$b_{n-1}$} \rput(4.1,1.6){$b_1$} \rput(-3.8,1.6){$b_n$}
\psdots[dotstyle=*,dotsize=2pt](-0.6,3.8)(0,3.8)(0.6,3.8)
\end{pspicture}
\caption{The arrangement of points.}
\label{fig:hardInstance}
\end{subfigure}%
\hfill
\begin{subfigure}[t]{0.4\textwidth}
\centering
\psset{unit=0.6,linewidth=0.015,dotsize=1.7pt}
\begin{pspicture}(-4,-1)(8,6)
\psline(0,0)(-3,4) \psline(0,0)(3,4) \psline(0,0)(0,5)
\psarc(0,0){5}{53.13}{126.87}
\psarc(0,0){4.87}{53.13}{126.87}
\psline(-4,5)(4,5) \psline(-4,4.87)(4,4.87) \psline(-2.922,3.896)(2.922,3.896)
\psdots[dotstyle=*, linecolor=red](-3,4)(0,5)(3,4)(-2.922,3.896)(0,4.87)(2.922,3.896)(0,3.896)
\rput(0,-0.5){$O$} \rput(0,5.5){$a_i$} \rput(3.7,4){$a_{i-1}$} \rput(-3.7,4){$a_{i+1}$} \rput(2.1,3.596){$b_{i-1}$} \rput (0.3,4.5){$b_i$} \rput(-2.1,3.596){$b_{i+1}$} \rput(0.3,3.596){$T$}
\end{pspicture}
\caption{The tangent lines at $a_i$ and $b_i$.}
\label{fig:tangentLineFig}
\end{subfigure}
\caption{An illustration of a hard instance to compute convex hull.}\label{fig:hard}
\end{figure}
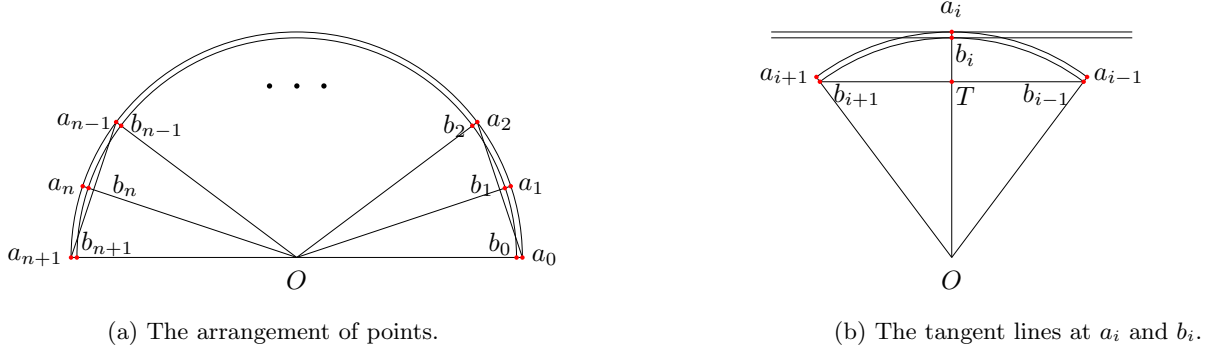

%% file: clower.tex
\section{Conditional Lower Bound}\label{sec:clower}

In this Section, we prove a conditional lower bound, higher than the unconditional one shown in Section~\ref{sec:lower} for small $h$. This conditional lower bound holds for those algorithms in the algebraic decision tree model, that is:

\begin{itemize}
	\item What algorithms can store in memory is a subset of input points. 
	\item The only operations that algorithms can perform is to test the sign of any continuous function evaluated on the points currently stored in memory.
\end{itemize}

Theorem~\ref{thm:constanthard} is our conditional lower bound, which implies that for constant $h$ any deterministic streaming algorithm that can compute the convex hull in the algebraic decision tree model requires $O(n^\delta)$ space, if the pass complexity $< h/2$. This lower bound is tight because the Gift-Wrapping algorithm can compute the convex hull using $O(h) = O(1)$ space and $h/2$ passes.

\begin{theorem}\label{thm:constanthard}
To compute the upper hull of $n$ points on a plane which has $h$ extreme points, any deterministic streaming algorithm in the algebraic decision tree model that uses $1/\delta$ passes for any $1/\delta < \min\{h/2, \log_h n\}$ requires $\Omega(n^\delta)$ storage of points.
\end{theorem}
\begin{proof}
Given two points $p_L$, $p_R$ and an open disk $D$ that lies above the line $p_L p_R$ and between two vertical lines $x = x_{p_L}$ and $x = x_{p_R}$, consider the problem that computes the upper hull of a $n$-point set $P \subseteq D$, where $P$ satisfies that $U(P \cup \{p_L, p_R\}) = U(P) \cup \{p_L, p_R\}$ and $U(P) = h$. We denote this problem as $S(p_L, p_R, D, n, h)$. The following lemma states that after each pass, the problem $S(p_L, p_R, D, n, h)$ remains as difficult as some problem $S(p^\prime_L, p^\prime_R, D^\prime, (n-2)/s, h-2)$. By setting $s \approx n^\delta$ and applying the lemma repeatedly, we have our theorem. 
\end{proof}

\begin{lemma}
Given a problem $S(p_L, p_R, D, n, h)$, for any deterministic streaming algorithm that has storage less than $s$ points, there exists a sequence of $n$ points $P$ inside $D$, a subset $X \subseteq P$, two points $p^\prime_L, p^\prime_R \in P$ and an open disk $D^\prime \subseteq D$, such that after we run the first pass of the algorithm, we have:
\begin{itemize}
	\item no point of $X$ is in memory but $p^\prime_L$ and $p^\prime_R$;
	\item the result of the pass would be identical if we move the points in $X$ to the arbitrary point in $D^\prime$;
	\item the upper hull of $X \cup \{p^\prime_L, p^\prime_R\}$ is equal to the upper hull of $P$;
	\item $|X| = \ceil{(n - 2)/ S} - 1$.
\end{itemize}
\end{lemma}
\begin{proof}
Consider $s$ points $p_1=(x_1, y_1), \ldots, p_s=(x_s, y_s)$ such that $p_1, \ldots, p_s \in D$, $p_1, \ldots, p_s$ and $p_L, p_R$ form a strictly concave chain, and for each $p_i$, there exists $q_i^L, q_i^R \in D$ such that $p_L, q_i^L, p_i, q_i^R, p_R$ forms the new upper hull above the concave chain $p_L, p_1, \ldots, p_s, p_R$. In other words, these five points are the upper hull of the point set $p_L, p_1, \ldots, p_s, p_R, q_i^L, q_i^R$. Let $U$ be the set containing all tuples $(x_1, y_1, \ldots, x_s, y_s)$ for such choices of $p_1, \ldots, p_s$. Note that $U$ is open and non-empty. To generate the first $n-2$ points in the stream, the adversary would choose the points only from $p_1, \ldots, p_s$. At every step, he picks $p_i$ such that it is currently not stored in the memory. Because the memory can only store $o(s)$ points, such $p_i$ always exists. By pigeonhole principle, there exists some $p_k$ such that it is chosen by the adversary $\ceil{(n-2)/s} - 1$ times. The adversary would stop to pick such $p_k$ when he is about to choose $p_k$ at $\ceil{(n-2)/s}$-th time but to choose any other $p_i$, no matter it is in memory or not. These $\ceil{(n-2)/s} - 1$ copies of $p_k$ are filled into $X$. Therefore $X$ is not stored in the memory and $|X| = \ceil{(n-2)/s} - 1$. 
To satisfy the second condition, during the execution of the algorithm, whenever a test is conducted, we consider the sign of the test function over all the possible tuples in $U$; if not all choices result the same sign, we refine $U$ to be a smaller open and non-empty set in which they do. Then, after the algorithm processes $n-2$ points, because $U$ is open, we can fix the choices of $p_1, \ldots, p_s$ and find an open disk $D_0$ around $p_k$, such that if any copy of $p_k$ is replaced by a point in $D_0$, the outcome of the tests is still the same. Finally, $q_k^L$ and $q_k^R$ are added into the stream as $p_L^\prime, p_R^\prime$. We refine $D_0$ further by finding a smaller open disk $D^\prime \subseteq D_0$ such that it is below two lines $p_L p_L^\prime$ and $p_R p_R^\prime$ and above the line $p_L^\prime p_R^\prime$. Therefore the third condition is satisfied and we are done.  
\end{proof}

%% file: proof.tex
\section{Omitted Proofs}\label{sec:proof}

\begin{proof}[Proof of Corollary \ref{cor:diam}]

For the diameter, we invoke two instantiations of convex hull W-stream algorithms simultaneously. One reports the extreme point starting at $p_\ell$ (leftmost) and the other starts at $p_r$ (rightmost). They both output the points in the clockwise order, i.e., one outputs the upper hull and the other output the lower one. After the execution of the algorithms,the first $s/2$ points starting from $p_\ell$ in the upper hull and first $s/2$ points starting from $p_r$ in the lower hull are loaded into the memory using one pass. Then we simulate Shamos' method by rotating the initial caliper $\overline{p_\ell p_r}$. When the $s/2$ points in the upper hull or the lower hull are processed, we using another pass to load the next $s/2$ points for the upper hull and the lower hull. We stop when all the extreme points on the convex hull are processed. Note that, it only requires $O(h/s)$ passes, $O(s)$ space and $O(h)$ time to do so. 

For the minimum enclosing rectangle, four instantiations of our W-stream algorithms are needed as we are rotating two orthogonal calipers. Then similar approach applies to execute the rotating calipers method detailed in~\cite{t83}. The complexities remain the same. This establishes Corollary~\ref{cor:diam}.
\end{proof}